\documentclass[a4paper,11pt]{article}
\usepackage[margin=1in]{geometry}  

\usepackage{listings}
\usepackage{amsmath,mathtools}
\usepackage{amsthm}
\usepackage{tikz}
\usepackage{caption,subcaption}
\usepackage{array}
\usepackage{mdwmath}
\usepackage{multirow}
\usepackage{mdwtab}
\usepackage{eqparbox}
\usepackage{multicol}
\usepackage{amsfonts}
\usepackage{multirow,bigstrut,threeparttable}
\usepackage{amsthm}
\usepackage{array}
\usepackage{bbm}
\usepackage{subcaption}
\usepackage{epstopdf}
\usepackage{mdwmath}
\usepackage{mdwtab}
\usepackage{eqparbox}
\usepackage{tikz}
\usepackage{latexsym}
\usepackage{amssymb}
\usepackage{bm}
\usepackage{amssymb}
\usepackage{graphicx}
\usepackage{mathrsfs}
\usepackage{epsfig}
\usepackage{psfrag}
\usepackage{setspace}
\usepackage{tikz-cd}

\usepackage[
CJKbookmarks=true,
bookmarksnumbered=true,
bookmarksopen=true,
colorlinks=true,
citecolor=red,
linkcolor=blue,
anchorcolor=red,
urlcolor=blue
]{hyperref}
\usepackage{cleveref}
\usepackage{algorithm}
\usepackage{algorithmic}
\usepackage{stfloats}
\usepackage[
autocite    = superscript,
backend     = bibtex, 
sortcites   = true,
style       = authoryear 
]{biblatex} 
\usepackage{nameref}
\usepackage[normalem]{ulem}

\usepackage{soul}

\input xy
\xyoption{all}

\theoremstyle{plain}
\newtheorem{theorem}{Theorem}[section]
\newtheorem{proposition}{Proposition}[section]
\newtheorem{lemma}{Lemma}[section]

\crefname{theorem}{Theorem}{Theorems}
\crefname{proposition}{Proposition}{Propositions}
\crefname{corollary}{Corollary}{Corollaries}

\newtheorem{definition}{Definition}[section]
\newtheorem{remark}{Remark}[section]
\newtheorem{example}{Example}[section]

\crefname{definition}{Definition}{Definitions}
\crefname{remark}{Remark}{Remarks}

\def\EE{\mathbb{E}}

\def\PP{\mathbb{P}}

\def\RR{\mathbb{R}}

\def\calN{\mathcal{N}}

\def\calS{\mathcal{S}}

\newcommand{\ind}{\textup{ind}}

\def\1{\mathbbm{1}}

\usepackage{booktabs}
\usepackage{adjustbox}
\usepackage{multirow}
\usepackage{listings}

\theoremstyle{plain}

\usepackage{xspace}

\definecolor{myblue}{rgb}{.8, .8, 1}
\definecolor{mathblue}{rgb}{0.2472, 0.24, 0.6} 
\definecolor{mathred}{rgb}{0.6, 0.24, 0.442893}
\definecolor{mathyellow}{rgb}{0.6, 0.547014, 0.24}

\usepackage{enumitem}

\newcommand{\FDR}{\text{FDR}}

\newcommand{\FDP}{\text{FDP}}

\newcommand{\BH}{{\text{BH}}}

\title{Simultaneous false discovery rate control in location families}

\author{
Zijun Gao\footnotemark[1]
\and
Wenjie Hu\footnotemark[2]
\and
Qingyuan Zhao\footnotemark[3]
}

\footnotetext[1]{Department of Data Science and Operations, University of Southern California, USA.}
\footnotetext[2]{Department of Biostatistics, Epidemiology \& Informatics, University of Pennsylvania, USA.}
\footnotetext[3]{Statistical Laboratory, University of Cambridge, UK.}

\begin{document}

\maketitle




\markboth{}{}





\begin{abstract}
  When testing a number of statistical hypotheses using data from
  location families, it is often useful to control the false discovery
  rate (FDR) not just for hypotheses of the null values but also of other parameter values that are deemed practically insignificant. Here we
  consider FDR as a curve indexed by the location parameter and
  suggest a simple generalization of the Benjamini-Hochberg procedure
  that controls the FDR curve below any user-specified level. As a
  corollary of our main result, we show that the standard Benjamini-Hochberg
  procedure---designed to control the FDR at the null---also provides simultaneous
  control of the whole FDR curve for free. We further demonstrate the implications of our results and some practical considerations with a numerical example.
 \end{abstract}


\section{Introduction}\label{sec:introduction}
The false discovery rate (FDR) has become a standard statistical error measure for
discovery in multiple hypothesis testing since the seminal work of
\cite{benjamini1995controlling}.
However, a p-value, or statistical significance, does not measure the
size of an effect or the importance of a result
\parencite{wasserstein2016asa}. In light of this, the concept of FDR may seem
restricted as it only concerns false discoveries at the null value; in
practice, it may also be desirable to have low FDR for small or
``practically insignificant'' values of the parameter. That is,
instead of just controlling the FDR at the null, it may be desirable
to control the whole ``FDR curve'', a concept we will formally
introduce later.

In this article we document a ``free lunch'' theorem that shows that the standard
Benjamini-Hochberg (BH) procedure, when applied to control the FDR at
the null, can also control the whole FDR curve at a level that
depends on the nominal FDR at the null. This observation is
based on a simple generalization of the BH procedure that can be
applied to control the whole FDR curve below any user-specified
level. Next, we will introduce the setup and the main theoretical
results before demonstrating the practical implications using a real dataset.

\section{Setup}\label{sec:setup}

Suppose we have independent, real-valued data $X_i \sim \PP_{\theta_i},
i=1,\dots,m$, from a location family $\{\PP_\theta: \theta \in
\mathbb{R}\}$, where the distribution function of $\PP_\theta$ is
given by $\PP_\theta(X \leq x) = F(x - \theta)$ and $F$ is continuous
and strictly increasing. We are interested in testing the one-sided hypotheses
\[
  H_{i,\theta}: \theta_i \geq \theta, \quad i=1,\dots,m,~\theta \in \RR.
\]
\clearpage
For example, if $\theta$ represents a treatment effect on a risk
factor, we may be only interested in discovering $\theta_i$ that is
sufficiently small (reduces the risk significantly). In this simple
setting, the p-value
for the hypothesis $H_{i,\theta}$ is given by \[P_{i,\theta} = F(X_i -
  \theta).\] So the test that rejects
$H_{i, \theta}$ when $P_{i, \theta} \leq \alpha$ has size $\alpha$.

\begin{definition}[FDR curve]
The false discovery proportion curve for a rejection set $\mathcal{S}
\subseteq \{1,\dots,m\}$ is given by
\begin{align*}
\FDP_{\mathcal{S}}(\theta)
:=
\frac{\sum_{i \in \mathcal{S}} \1_{\{\theta_i \geq \theta\}} }{\max
\{1, |\mathcal{S}|\}}, \quad \theta \in \mathbb{R}.
\end{align*}
The FDR curve of a multiple testing procedure that rejects a random
subset $\mathcal{S}(X_1,\dots,X_m)$ is then defined as
$\FDR(\theta)=\EE[\FDP_{\mathcal{S}(X_1,\dots,X_m)}(\theta)]$.
\end{definition}

To test $H_{1,0},\dots,H_{m,0}$, the standard $\BH_q$ procedure computes
\[
  i^* = \max \left\{i: P_{(i), 0} \leq i q / m, 1 \leq i \leq m\right\},
\]
where $P_{(1), 0} \leq \dotsc \leq P_{(m), 0}$ are the order
statistics of the p-values for the null hypotheses and $q$ is a
user-specified value between 0 and $1$. The $\BH_q$ procedure then
rejects all hypotheses with p-values less than or equal to
$P_{(i^*)}$:
\[
  \mathcal{S}_{\BH_q}(X_1,\dots,X_m) = \{i: P_{i,0} \leq P_{(i^*)}\}.
\]
It is not difficult to see that this is the largest subset
$\mathcal{S} \subseteq \{1,\dots,m\}$ such that $i \in \mathcal{S}$ if
and only if $P_{i,0} \leq |\mathcal{S}|q/m$.
It has been shown that this standard $\BH_q$ procedure controls the FDR at
the null in the sense
that $\FDR(0) \leq q$, or equivalently, controls the FDR
curve below a simple step function:
\begin{equation} \label{eq:q-bh}
  \FDR(\theta) \leq q_{\BH}(\theta) := q \cdot \1_{\{\theta \ge 0\}} +
  \1_{\{\theta < 0\}}, \quad \theta \in \RR.
\end{equation}
See \cite{benjamini2001control,benjamini2005false}, \cite{storey2004strong}, and \cite{ferreira2006} for different
ways to prove this celebrated result.

\section{Main Results}

Our main methodological contribution is a generalized $\BH_q$
procedure for FDR curve control. Instead of considering $q \in \RR$ as a nominal FDR
level at the null, let $q: \mathbb{R} \to [0,1]$ be a user-specified nominal
FDR curve.
For example, to control the FDR at a finite number of locations specified by $(\theta_j, q_j), j=1,\dots, l$, it suffices to use the step function
\[
q(\theta) := \inf_{\theta_j \le \theta} q_j, \quad \theta \in \mathbb{R}, \quad \text{and~} \inf_{j \in \emptyset} q_j = 1.
\]
Without loss of generality, we will assume $q$ is
non-increasing---the target FDR can only become smaller as $\theta$
increases, as $H_{i,\theta}$ is stronger than $H_{i,\theta'}$ when $\theta > \theta'$ in our setup.

Consider the following generalized $\BH_q$ procedure:
\begin{enumerate}
    \item Compute the following target FDR curve-normalized ``p-values'' (can be larger than $1$):
\begin{align}\label{eq:p.value.FDR.curve}
   \bar{P}_{i} = \sup_{\theta \in \RR} P_{i,\theta}/q(\theta),
  \quad i=1,\dots,m.
 \end{align}
 \item Apply the original $\BH_q$ procedure to
$\bar{P}_{1},\dots,\bar{P}_m$ with $q = 1$.
\end{enumerate}
In other words, the discovery set $\mathcal{S}_{\BH_q}(X_1,\dots,X_m)$ found by the generalized $\BH_q$ procedure is the largest subset
 $\mathcal{S} \subseteq \{1,\dots,m\}$ such that $i \in \mathcal{S}$ if
 and only if $\bar{P}_i \leq |\mathcal{S}|/m$.
It is easy to see that this reduces to the original $\BH_q$
procedure when the step function $q_{\BH}$ in \eqref{eq:q-bh} is
used as the target FDR curve. For this reason, we will still refer to
this generalization as the $\BH_q$ procedure with the
understanding that $q$ can now be a real-valued function.

Before stating our main result, let us  introduce a modification to the
user-specified curve $q$:
\begin{equation}
  \label{eq:q-star}
  q^*(\theta) = \inf_{\theta' \in \RR}
    \sup_{a \leq x - \theta' \leq b} q(\theta') \cdot
    \frac{F(x - \theta)}{F(x - \theta')}, \quad \text{where}~a =
    F^{-1}(q(\theta')/m) ~\text{and}~b = F^{-1}(q(\theta')).
\end{equation}
Throughout, $F^{-1}(u)=\inf\{x:F(x)\ge u\}$ denotes the generalized quantile,
with $F^{-1}(1)=+\infty$. It is easy to see that $q^*(\theta) \leq q(\theta)$ is always true by
taking $\theta' =
\theta$ in the definition above.

\begin{remark}
    The choice of $(a,b)$ ensures that $q(\theta') / m \leq F(x
- \theta') \leq q(\theta')$, so the supremum in \eqref{eq:q-star} is taken over all values
of $x$ that have a p-value against $\theta'$ in $[q(\theta')/m,
q(\theta')]$, the maximum range of rejection threshold for the standard
$\BH_{q(\theta')}$ procedure.
If instead we simply use $a=-\infty$ and $b=\infty$, it will typically be the case that $q^*(\theta) = q(\theta)$ for all $\theta$. (The ratio ${F(x - \theta)}/{F(x - \theta')}$ typically converges to $1$ as $x \to \infty$ for $\theta>\theta'$ and to $\infty$ as $x \to -\infty$ for $\theta<\theta'$.)
\end{remark}

\begin{theorem} \label{thm:main}
  The generalized $\BH_{q}$ procedure simultaneously controls
  the FDR curve in the following sense
  \begin{equation} \label{eq:thm:main}
    \sup_{\theta \in \RR} \frac{\FDR(\theta)}{q^*(\theta)} \le
    \EE\left[\sup_{\theta \in \RR}
      \frac{\FDP(\theta)}{q^*(\theta)}\right] \le 1.
  \end{equation}
  So at any location $\theta$ (possibly post hoc) we have
  $\FDR(\theta) \leq q^{*}(\theta)$ (almost surely).
\end{theorem}

\Cref{thm:main} has two important practical implications:
\begin{enumerate}
\item the standard
$\BH_q$ procedure (applied to null p-values only) may offer FDR control
at other locations for free,
\item and the generalized $\BH_q$ procedure
may provide FDR curve control that is stronger than the user-specified
curve $q$.
\end{enumerate}
The first implication is really just a special case of the second if
one takes the
viewpoint that the standard $\BH_q$ procedure corresponds to
using a simple step function $q_{\BH}$ in \eqref{eq:q-bh} as the
target FDR curve. In this case, it can be shown that
$q_{\BH}^{*}(0) = q_{\BH}(0)$; see the remark after
\Cref{lemm:minimizer} in the Appendix. Geometrically, this means
that $q^{*}_{\BH}$ is below $q_{\BH}$
but the two functions touch at $\theta = 0$; see
\Cref{fig:qcLowerBound} below for an illustration.

The following result shows that the observation above is not a coincidence: under some mild conditions, $q^*$ and
$q$ always touch at least one point.
We say the location family $F(x - \theta)$ satisfies the \emph{monotone ratio property} if
\begin{align}\label{cond:monotone.ratio}
    \text{${F(x-\theta')}/{F(x-\theta)}$
is monotone in $x$ for any $\theta$, $\theta' \in \mathbb{R}$.}
\end{align}
This property ensures that the supremum in \eqref{eq:q-star} is attained at one of the endpoints, that is $x-\theta'=a$ or $x-\theta'=b$.
It holds for many common location families, including the Gaussian location model (i.e.\ $F = \Phi$ is the CDF of the standard normal distribution); more generally, it holds whenever $F$ admits a density $f$ such that the ratio $f(x)/F(x)$ is monotone in $x$.

\begin{proposition}\label{prop:touching.point}
Suppose $F$ is continuous, strictly
increasing, and has the monotone ratio
property~\eqref{cond:monotone.ratio}. If $q$ is decreasing,
right-continuous, bounded away from zero, and $q(\underline{\theta}) =
1$ for some $\underline{\theta}>-\infty$, then there exists $\theta^*
\in \RR$
such that $q^*(\theta^*)=q(\theta^*)$.
\end{proposition}


\Cref{prop:touching.point} shows that, even for general $q$ functions, it is impossible that $q^{*}(\theta) < q(\theta)$ for
\emph{all} $\theta \in \RR$, that is, the actual FDR control of the
generalized $\BH_q$ procedure cannot be uniformly better than the
nominal FDR control. See \Cref{fig:qcLowerBound,fig:qcBreastCancerSNR}
for some examples.

An interesting question is how tight the inequalities in
\eqref{eq:thm:main} are. The following lower bound of the worst-case
FDR curve can be obtained by setting $\theta_1 = \dots = \theta_m =
\theta$.

\begin{proposition}[Lower bound]\label{prop:lower.bound}
Suppose $F$ is continuous, strictly increasing, and has the monotone ratio
   property~\eqref{cond:monotone.ratio}. Let $q^*$ be the transformed
   FDR curve for the standard $\BH_q$ procedure obtained by applying
   \eqref{eq:q-star} to \eqref{eq:q-bh}. Then
   \[
     \sup_{\theta_1,\dots,\theta_m \in \RR} \FDR(\theta) \geq 1 - e^{-q^*(\theta)} \quad \text{for all}~\theta < 0.
   \]
\end{proposition}



When $\theta < 0$ and $q^*(\theta)$ is small, the lower bound $1 -
e^{-q^*(\theta)}$ above is close to $q^*(\theta)$, indicating that our
analysis is not too loose. This is further confirmed by the example below.

\begin{example}
Consider the Gaussian location family $X_i \sim \calN(\theta_i,1)$ so $F = \Phi$.
Figure~\ref{fig:qcLowerBound} shows that $q^*$ touches $q_{\BH}$ at $\theta = 0$, and is substantially smaller than $q_{\BH}$ for all $\theta \neq 0$ displayed, suggesting that the BH procedure can provide much stronger control of the FDR curve than the simple step function $q_{\BH}$.
We further add to the Figure a Monte Carlo estimate of the FDR
curve when $\theta_1 = \dots = \theta_m = \theta$ (blue dashed), which
lower bounds the worst-case FDR curve over all possible configurations
of $(\theta_1,\dots,\theta_m)$.
\end{example}

\begin{figure}[btp]
    \centering
    \begin{subfigure}[t]{0.32\linewidth}
        \centering
        \includegraphics[width=\linewidth,trim=0 0.45cm 0 0,clip]{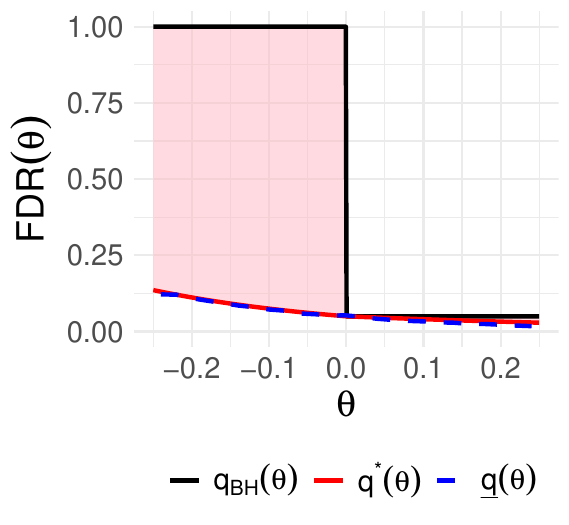}
        \caption*{\qquad\qquad (a) $q(0)=0.05$}
    \end{subfigure}
    \begin{subfigure}[t]{0.32\linewidth}
        \centering
        \includegraphics[width=\linewidth,trim=0 0.45cm 0 0,clip]{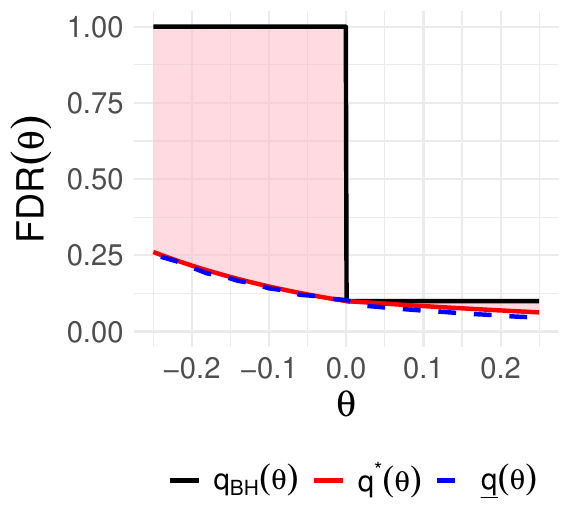}
        \caption*{\qquad \qquad (b) $q(0)=0.1$}
    \end{subfigure}
    \begin{subfigure}[t]{0.32\linewidth}
        \centering
        \includegraphics[width=\linewidth,trim=0 0.45cm 0 0,clip]{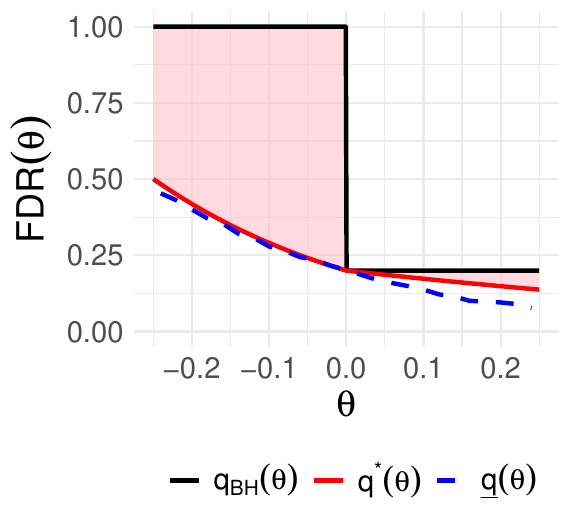}
        \caption*{\qquad \qquad (c) $q(0)=0.2$}
    \end{subfigure}
    \caption{FDR curve control for Gaussian location model. The curve
    $\underline{q}(\theta)$ is a Monte Carlo estimate of the FDR curve
    when $\theta_1 = \dots = \theta_m = \theta$.}
    \label{fig:qcLowerBound}
\end{figure}


\section{Generalizations beyond a Single Location Family and a Real Data Example}
\label{sec:numerical}

In many applications, data for different hypotheses might not be generated from a single location family, but many results obtained above can be suitably generalized.
Here we consider the problem with independent $X_1,\dots,X_m$ that satisfy
$\PP_{i,\theta_i}(X_i \le x)=F_i(x-\theta_i)$, $i=1,\dots,m$, where each $F_i$ is continuous and strictly increasing.
The p-value for testing $H_{i,\theta}:\theta_i \ge \theta$ is then $P_{i,\theta}=F_i(X_i-\theta)$,
and the FDR-curve normalized p-value $\bar{P}_i$ remains as in \eqref{eq:p.value.FDR.curve}.
We update the definition of $q^*$ in \eqref{eq:q-star} as
\begin{equation}
\label{eq:q-star.heterogeneous.CDF}
q^*(\theta)
:=
\inf_{\theta' \in \mathbb R}
\sup_{i\in[m]}
\sup_{a_i \le x-\theta' \le b_i}
q(\theta') \cdot \frac{F_i(x-\theta)}{F_i(x-\theta')},
\quad
a_i = F_i^{-1}\!\left({q(\theta')/m}\right),
~~
b_i = F_i^{-1}\left(q(\theta')\right),
\end{equation}
where we further take the supremum over $F_i$ and $[m] = \{1,\dots,m\}$.
When $F_i = F$ for all $i$, this reduces to \eqref{eq:q-star}.
As before, $q^*(\theta)\le q(\theta)$ by taking $\theta'=\theta$.


The FDR control result in \Cref{thm:main} continues to hold in this setting; see \Cref{appendix:proofs} for the proof.
However, \Cref{prop:touching.point} may not hold because of the additional supremum.
Nevertheless,  \Cref{prop:touching.point.heterogeneous.CDF} below shows that when $q$ is piecewise constant, there always exists a subset of the original constraints whose induced $q^*$ curve still controls all prescribed FDR constraints and touches the target $q$ curve at least once. We will further discuss this using a numerical example; see \Cref{fig:qcBreastCancerEffectSize}.

\begin{proposition}\label{prop:touching.point.heterogeneous.CDF}
Suppose $q$ is piecewise constant (let $\Theta\subset\mathbb R$ denote the set of jump points), decreasing, and right-continuous.
For any subset $\Theta' \subseteq \Theta$, let $q^*_{\Theta'}$ denote the curve obtained from \eqref{eq:q-star.heterogeneous.CDF} with the infimum restricted to $\theta' \in \Theta'$. Then there exists a nonempty subset $\Theta' \subseteq \Theta$ such that
$q^*_{\Theta'}(\theta) \le q(\theta)$ for all $\theta \in \RR$ and $q^*_{\Theta'}(\theta')=q(\theta')$ for some $\theta' \in \Theta'$.
\end{proposition}




We provide a numerical illustration comparing formulating the FDR curve as a function of effect size (different $F_i$) vs.\
signal-to-noise ratio (shared $F_i = F$) using the benchmark breast cancer dataset in \cite{Storey2003}.
This dataset contains expression measurements for $m=3,170$ genes from hereditary breast tumors with BRCA1 and BRCA2 mutations.
Suppose we are interested in finding genes with substantially higher BRCA2 expression  than BRCA1 expression. We assume that the difference in mean log  expression levels (base 2) for gene $i = 1,\dots,m$ satisfies $X_i \sim \mathcal{N}(\mu_i, \sigma_i^2)$. Consider two sets of hypotheses below.
\begin{enumerate}
    \item Effect sizes: \(H_{i,\mu}: \mu_i \geq\mu\), using the test statistic $X_i$ and \(F_i(x)=\Phi(x/\hat{\sigma}_i)\), which varies across genes through the estimated standard deviation \(\hat{\sigma}_i\).


\item Signal-to-noise ratio: \(H_{i,\theta}: \theta_i = \mu_i/\sigma_i \geq \theta\), using test statistic $X_i/\hat{\sigma}_i$ and a shared \(F(x)=\Phi(x)\).
\end{enumerate}

\begin{figure}[tbp]
    \centering
    \includegraphics[width=0.24\linewidth]{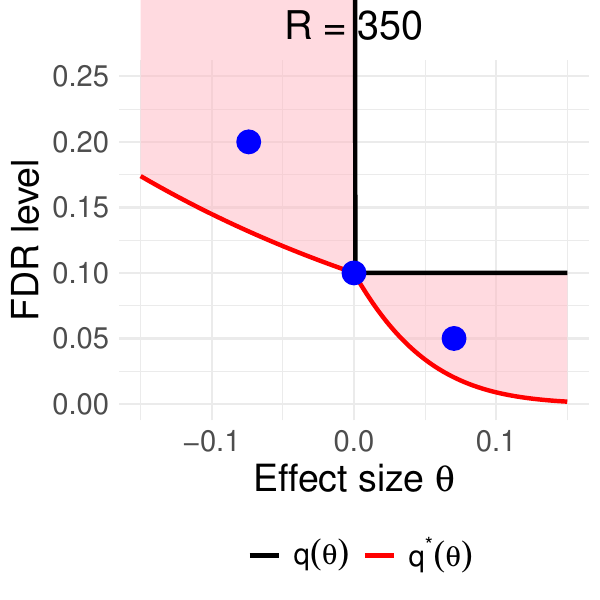}
    \includegraphics[width=0.24\linewidth]{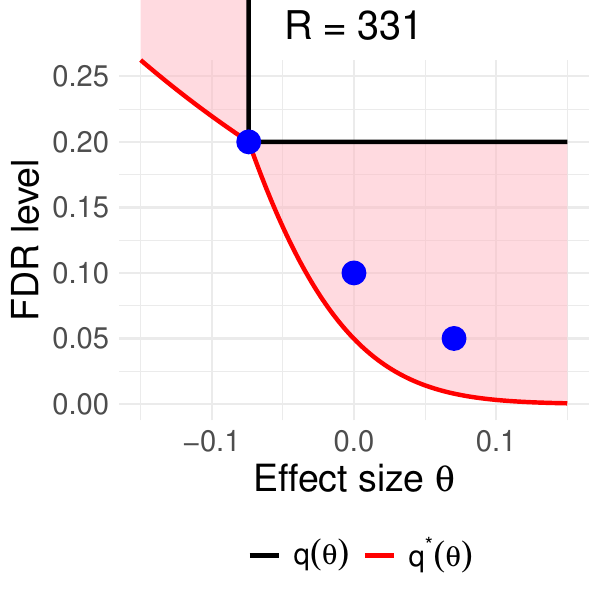}
    \includegraphics[width=0.24\linewidth]{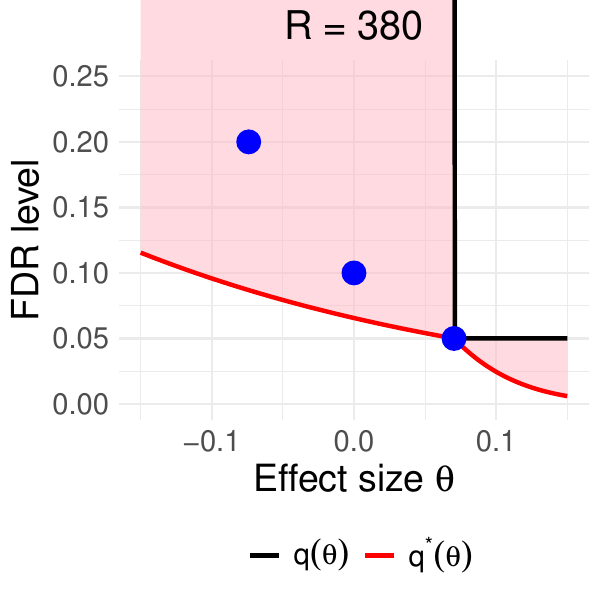}
    \includegraphics[width=0.24\linewidth]{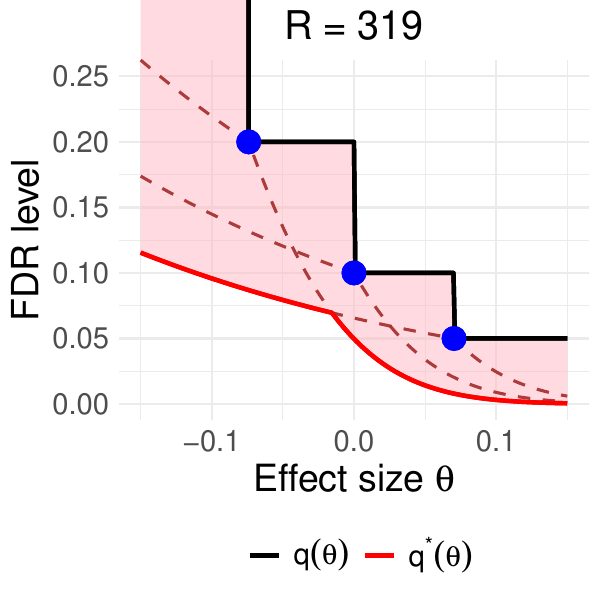}

  \caption{FDR curve control for breast cancer dataset testing the effect size with hypotheses \(H_{i,\mu}: \mu_i  \geq \mu\) with different $F_i$. $R$ denotes the number of rejections produced by the corresponding generalized BH procedure.}
\label{fig:qcBreastCancerEffectSize}
\end{figure}

\begin{figure}[tbp]
    \centering
    \includegraphics[width=0.24\linewidth]{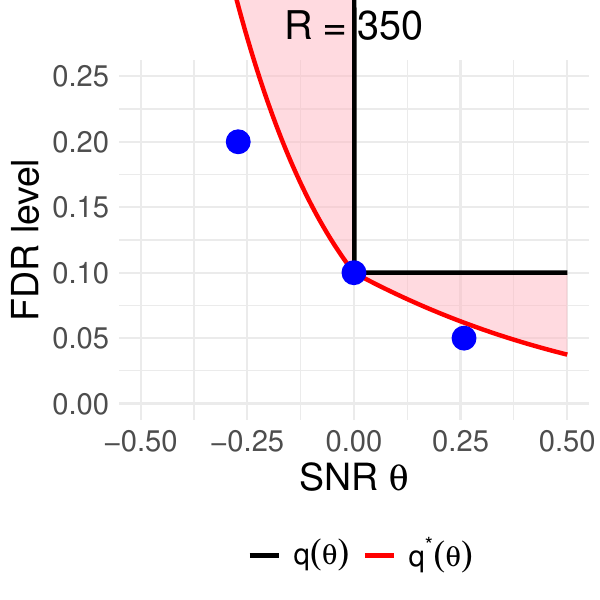}
    \includegraphics[width=0.24\linewidth]{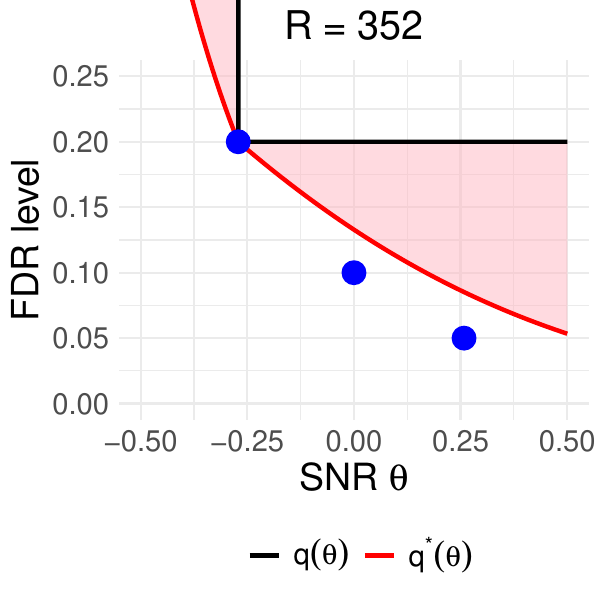}
    \includegraphics[width=0.24\linewidth]{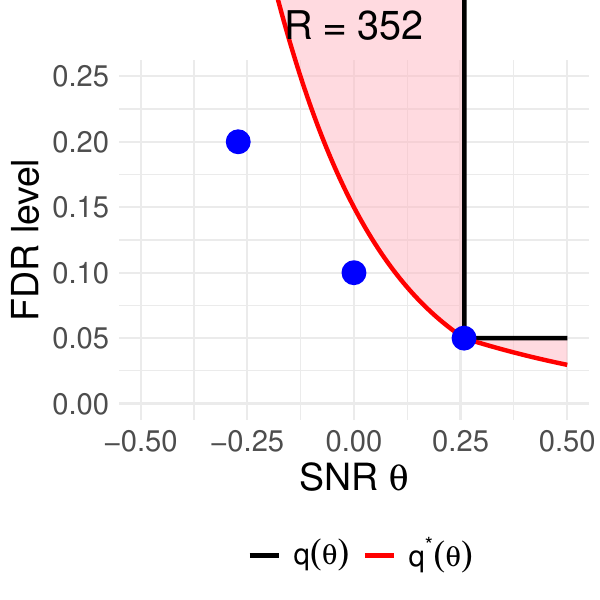}
    \includegraphics[width=0.24\linewidth]{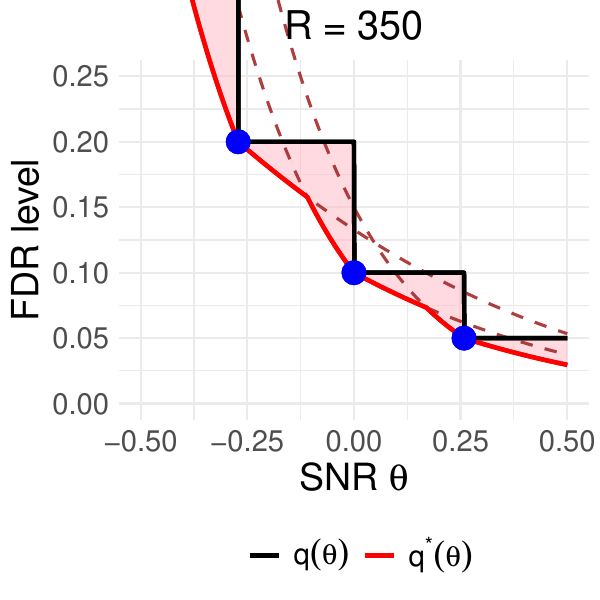}

    \caption{FDR curve control for breast cancer dataset testing the signal-to-noise ratio with hypotheses \(H_{i,\theta}: \theta_i = \mu_i/\sigma_i \geq \theta\)  with shared $F$. $R$ denotes the number of rejections produced by the corresponding generalized BH procedure. }
    \label{fig:qcBreastCancerSNR}
\end{figure}

For effect-size testing, we adopt FDR constraints of $\mu = 0$ at level $q = 0.1$, $\mu = 0.07$ (a $5\%$ decrease) at level $q = 0.05$, and $\mu = -0.07$ (a $5\%$ improvement) at level $q = 0.2$.
As mentioned above, the conclusion in \Cref{prop:touching.point} may not apply here because the distribution function $F_i$ varies across the hypotheses.
We observe that controlling any single FDR constraint is sufficient to control the whole FDR curve below $q$, and taking the infimum over the $q$ specified by all three constraints makes the $q^*$ curve strictly below $q$ and yields fewer rejections (319 compared to 350, 331, 380).
Therefore, instead of using all the FDR constraints, in this setting we recommend using fewest possible constraints as the target FDR curve $q$ such that $q^*$ is below $q$, which always exists by \Cref{prop:touching.point.heterogeneous.CDF}. (Note that $q^*$ can be computed without the data, thus selecting $q^*$ incurs no selection bias.)


For signal-to-noise ratio testing, we use the constraints $q(-0.27) = 0.2$, $q(0) = 0.1$, $q(0.26) = 0.05$, which roughly corresponds to the effect size levels used above (the median of $\hat{\sigma}_i$ is around $0.27$ and $\log_2(1.05) / 0.27 = 0.2607$). In this case, the conclusion of \Cref{prop:touching.point} applies and can be observed from \Cref{fig:qcBreastCancerSNR}.
In fact, none of these FDR constraints implies another, and the $q^*$ obtained using all three constraints touches $q$ at all three locations.
The resulting rejection sets are quite stable across different choices of $q$ (making 350 to 352 rejections).

\section{Further Discussion}

The approach to simultaneous FDR control in this paper is based on FDR curve-normalization of p-values in \eqref{eq:p.value.FDR.curve}. We refer the reader to \cite{gao2024constructive} for a more general framework that can be applied to other selective compound decision problems (e.g.\ with hybrid error criteria combining errors made at different locations) beyond multiple hypothesis testing.


There is potential for adopting an empirical Bayes approach to select the FDR curve $q$ using an estimated prior distribution of the location parameter. Heuristically, if $\theta_1,\dots,\theta_m$ are drawn independently from a distribution $\pi$, then the false discovery proportion curve for the rejection rule $\mathcal{S} = \{i: X_i \leq x\}$ can be approximated by
\[
  \FDP_\mathcal{S}(\theta) \to \PP(\theta_i \geq \theta \mid X_i \leq x) = \frac{\int_{\theta}^\infty \pi(t) F(x - t) dt}{\int_{-\infty}^\infty \pi(t) F(x - t) dt}, \quad \text{as}~m \to \infty.
\]
Thus, given an estimate of $\pi$ by deconvolving the empirical
distribution of $X_1,\dots,X_m$ \parencite[see e.g.][]{efron2014two} and
a desired number of rejections, one could use the formula above to
compute a target FDR curve. Further careful treatment would be needed
to address the double-dipping issue to guarantee FDR control.

\section*{Acknowledgement}
QZ was partly supported by the EPSRC grant EP/V049968/1.

\printbibliography

\appendix


\section{Proofs}

\label{appendix:proofs}

We maintain the assumption that $F$ (or $F_1,\dots,F_m$ in the general
case considered in \Cref{sec:numerical}) is continuous and strictly increasing.


\begin{proof}[Proof of \Cref{thm:main}]
Here we prove an extended version of \Cref{thm:main}: the inequalities
in \eqref{eq:thm:main} still hold in the more general setting
described at the beginning of \Cref{sec:numerical} with $q^*$ defined
in \eqref{eq:q-star.heterogeneous.CDF}.

  The first inequality in \eqref{eq:thm:main} directly follows from the definition $\FDR(\theta) =
  \EE[\FDP(\theta)]$. In proving the second inequality in \eqref{eq:thm:main} we will write the rejection
  set $\mathcal{S}_{\BH_q}(X_1,\dots,X_m)$ with the target FDR
  curve $q$   as $\mathcal{S}_q$. Using the assumption that $q$ is
  non-increasing, we have
  \begin{align*}
    \sup_{\theta \in \RR}\frac{\FDP(\theta)}{q(\theta)} &= \sup_{\theta
      \in \RR} \frac{\sum_{i=1}^m \1_{\{\theta_i
        \geq \theta, i\in \mathcal{S}_q\}} }{q(\theta) \max\{1,
                                                          |\mathcal{S}_q|\}} \\
    & \leq \sum_{i=1}^m \sup_{\theta \in \RR} \frac{\1_{\{\theta_i
        \geq \theta, i\in \mathcal{S}_q\}} }{q(\theta) \max\{1,
                                                          |\mathcal{S}_q|\}}
    \\
    &= \sum_{i=1}^m \frac{\1_{\{ i\in \mathcal{S}_q\}} }{q(\theta_i) \max\{1,
                                                          |\mathcal{S}_q|\}}.
  \end{align*}
 By adapting the standard ``leave-one-out'' argument \parencite[see e.g.][]{wang2022elementary}, it can be shown that
  \[
  \EE \left[ \sup_{\theta \in \RR}\frac{\FDP(\theta)}{q(\theta)}
  \right] \leq 1.
  \]
  The proof is deferred to Lemma~\ref{lemma:fdr control} below. A more
  general proof strategy can be found in \cite{gao2024constructive}.

  To complete the proof of \Cref{thm:main}, it suffices to show that
  $\mathcal{S}_{q} = \mathcal{S}_{q^{*}}$, that is, using the more
  stringent FDR curve $q^{*}$ as target always produces the same
  rejection set as just using $q$. Let
  \[
    \bar{P}^{*}_{i} = \sup_{\theta \in \RR} P_{i,\theta}/q^{*}(\theta), i=1,\dots,m,
  \]
  be the normalized p-values by the FDR curve $q^{*}$. Since $q^*(\theta)
  \le q(\theta)$, we immediately have $\bar{P}_{i} \le
   \bar{P}^\ast_{i}$ and $\calS_{q^{*}} \subseteq
   \calS_{q}$. To prove $\calS_{q^{*}} \supseteq
   \calS_{q}$, it suffices to show $\bar{P}_i^{*} \leq
   |\mathcal{S}_q|/m$ for all $i \in \calS_q$. Without loss of
   generality, assume $|\mathcal{S}_q| > 0$, so for any $i \in
   \mathcal{S}_q$, we have
   \begin{equation}
     \label{eq:help}
     F_i(X_i - \theta) / q(\theta) = P_{i,\theta} / q(\theta) \leq
     \bar{P}_i \leq |\mathcal{S}_q|/m \leq 1, \quad \text{for all}~\theta \in \RR.
   \end{equation}
   Therefore
   \begin{align*}
     \bar{P}_i^{*} &= \sup_{\theta \in \RR} \frac{F_i(X_i -
     \theta)}{\inf_{\theta' \in \RR} \sup_{j\in \{1,\ldots,m\}}    \sup_{a_j \leq x - \theta' \leq b_j}
    q(\theta') \cdot \frac{F_j(x - \theta)}{F_j(x - \theta')}} \\
    &=\sup_{\theta' \in \RR} \sup_{\theta \in \RR} \frac{F_i(X_i -
     \theta)}{ \sup_{j\in \{1,\ldots,m\}} \sup_{a_j \leq x - \theta' \leq b_j} q(\theta') \cdot
    \frac{F_j(x - \theta)}{F_j(x - \theta')}} \\
    &\leq \sup_{\theta' \in \RR} \sup_{\theta \in \RR}
      \frac{F_i(\max\{X_i,a_i+\theta'\} - \theta)}{ q(\theta') \cdot
      \frac{F_i(\max\{X_i,a_i+\theta'\} - \theta)}{F_i(\max\{X_i,a_i+\theta'\}
      - \theta')}} \\
    & = \sup_{\theta' \in \RR} \frac{F_i(\max\{X_i,a_i+\theta'\}
      - \theta')}{q(\theta')} \\
   &= \sup_{\theta' \in \RR} \max\left\{
     \frac{P_{i,\theta'}}{q(\theta')}, \frac{1}{m}
     \right\} \leq |\mathcal{S}_q|/m.
   \end{align*}
   The first inequality above follows from \eqref{eq:help}, $q(\theta)
   \leq 1$ for all $\theta$, and the monotonicity of $F$. This
   completes the proof.
\end{proof}

\begin{lemma}\label{lemma:fdr control}
    For the generalized $\BH_q$ procedure, we have
    \[
    \EE\left\{ \frac{ \1_{\{i\in \mathcal{S}_q\}}  }{q(\theta_i) \max\{1,|\mathcal{S}_q|\}} \right\}  \le \frac{1}{m}, \text{~for~} i=1,\ldots,m.
    \]
\end{lemma}

\begin{proof}[Proof of \Cref{lemma:fdr control}]
Recall that the generalized $\BH_q$ procedure operates by applying the
standard BH procedure to the normalized
    $p$-values $\bar{\bm P} = (\bar{P}_1, \dots, \bar{P}_m)$. Let $\bm
    S(\bar{\bm P})=(S_1, \ldots, S_m) \in \{0,1\}^m$ denote its
    selection indicator vector. Define the leave-one-out vector $\bar{\bm P}^{-i}$ by removing $\bar{P}_i$. We construct a leave-one-out selected set $\calS_q^{-i}$ by manually setting the $i$-th normalized p-value to $0$. That is, its selection indicator vector is defined as $\bm S(\bar{P}_1, \dots, \bar{P}_{i-1}, 0, \bar{P}_{i+1}, \dots, \bar{P}_m)$.
    A useful property of the standard BH procedure is that it is a \emph{simple selection rule} \parencite{benjamini2014selective,gao2024constructive}: if hypothesis $i$ is selected, changing $\bar{P}_i$ to $0$ (making it even more significant) will not alter the total number of rejections or the selection status of any other hypotheses.
    This property guarantees that whenever $i \in \calS_q$ (i.e., $S_i(\bar{\bm P}) = 1$), the selected set remains exactly the same, which implies $|\calS_q| = |\calS_q^{-i}|$.
    Furthermore, the leave-one-out selected set $\calS_q^{-i}$ depends
    only on $\bar{\bm P}^{-i}$, so $\calS_q^{-i} \ind
    P_{i,\theta_i}$. Then we have
    \begin{equation*}
    \begin{aligned}
            \EE\left\{ \frac{ \1_{\{i\in \mathcal{S}_q\}}
            }{q(\theta_i) \max\{1,|\mathcal{S}_q|\}} \right\}
            =&~\EE\left\{ \frac{ \1_{\{i\in \mathcal{S}_q, P_{i, \theta_i} \leq q(\theta_i) |\mathcal{S}_q|/m\}}  }{q(\theta_i) \max\{1,|\mathcal{S}_q|\}} \right\}  \\
            =&~ \EE\left\{ \frac{ \1_{\{i\in \mathcal{S}_q, P_{i,
            \theta_i} \leq q(\theta_i) |\mathcal{S}_q^{-i}|/m\}}
            }{q(\theta_i) \max\{1,|\mathcal{S}_q^{-i}|\}} \right\} \\
            \leq&~ \EE\left\{ \frac{ \1_{\{P_{i,
            \theta_i} \leq q(\theta_i) |\mathcal{S}_q^{-i}|/m\}}
            }{q(\theta_i) \max\{1,|\mathcal{S}_q^{-i}|\}} \right\} \\
           = &~ \EE \left\{ \frac{q(\theta_i) |\mathcal{S}_q^{-i}|/m}{q(\theta_i) \max\{1,|\mathcal{S}_q^{-i}|\}} \right\} \le  \frac{1}{m}. \\
    \end{aligned}
    \end{equation*}
    This completes the proof of \Cref{lemma:fdr control}.
\end{proof}



Next, we provide a proof of \Cref{prop:touching.point}. For a pair $(t, q(t)) \in \RR \times [0,1]$, define
\begin{align}
    \bar{q}_t(\theta) := q(t) \cdot \1_{\{\theta \ge t\}} + \1_{\{\theta <
      t\}},
\end{align}
which generalizes the simple step function $q_{\BH}$ in \eqref{eq:q-bh}. Let $\bar{q}^*_t(\theta)$ be the corresponding $q^*$ function obtained by the transformation in \eqref{eq:q-star}.

\begin{lemma}\label{lemm:minimizer}
   For all $t \in \RR$,
\begin{align*}
\bar{q}^*_t(\theta)
=
\min\left\{
\bar{q}_t(\theta),\,
\sup_{a \le x-t \le b}
\bar{q}_t(t)\frac{F(x-\theta)}{F(x-t)}
\right\}
=
\min\left\{
1,\,
\sup_{a \le x-t \le b}
\bar{q}_t(t)\frac{F(x-\theta)}{F(x-t)}
\right\}.
\end{align*}
In particular, if $\bar{q}_t(t)=1$, then $\bar{q}^*_t(\theta) = 1$ for all $\theta$.
\end{lemma}

The proof of \Cref{lemm:minimizer} is deferred. By taking $\theta = t$
in \Cref{lemm:minimizer} and using the definition of $\bar q_t(t)$, we
immediately find
\begin{align*}
        \bar{q}^*_t(t) =  q(t) \quad \text{for all}~t \in \RR.
\end{align*}
That is, the modified curve touches the original curve at $t$.

The following definition is useful in studying the relation between $q$ and $q^*$.

\begin{definition}\label{defi:order}
   For any $\theta_1,\theta_2 \in \RR$ and $q_1,q_2 \in (0,1]$, we write $(\theta_1,q_1) \preceq (\theta_2,q_2)$ iff $\bar q_{\theta_2}^{*}(\theta_1) \le q_1$, and $(\theta_1,q_1) \prec (\theta_2,q_2)$ iff $\bar q_{\theta_2}^{*}(\theta_1) < q_1$.
   We write $(\theta_1,q_1) \sim (\theta_2,q_2)$ iff $\bar q_{\theta_1}^*(\theta) = \bar q_{\theta_2}^*(\theta)$ for all $\theta$.
\end{definition}

\begin{lemma}\label{lemm:equivalence}
    If $F$ satisfies the monotone ratio property~\eqref{cond:monotone.ratio}, then  $(\theta_1, q_1) \preceq (\theta_2, q_2)$  is true if and only if \begin{align}\label{eq:equivalence}
        F^{-1}(q_1) + \theta_1 \ge F^{-1}(q_2) + \theta_2,~\text{and}~
        F^{-1}(q_1/m) + \theta_1 \ge F^{-1}(q_2/m) + \theta_2.
    \end{align}
    Furthermore, $(\theta_1, q_1) \prec (\theta_2, q_2)$ is true if and only if the inequalities in \eqref{eq:equivalence} hold and at least one of them is strict.
\end{lemma}

The proof of \Cref{lemm:equivalence} is also deferred. Next we prove \Cref{prop:touching.point} in the main text.

\begin{proof}[Proof of \Cref{prop:touching.point}]
Denote $A(\theta)=\theta+F^{-1}(q(\theta)/m)$, $B(\theta)=\theta+F^{-1}(q(\theta))$, and $S(\theta)=A(\theta)+B(\theta)$.
As $q$ is decreasing and $q(\underline{\theta}) = 1$, we only need to consider constraints corresponding to $\theta >\underline{\theta}$.
Since $q(\theta)$ is bounded away from zero, $S(\theta)$ is bounded below for $\theta >\underline{\theta}$ and $S(\theta) \rightarrow +\infty$ when $\theta \rightarrow +\infty$. Define $S_*:=\inf_\theta S(\theta) =\inf_{\theta > \underline{\theta}} S(\theta) > -\infty$.
Then there exists a  $\theta^*$ such that $S(\theta^*) = S_*$. (This
is because there exists a bounded sequence $\{\theta_n\}$ (not necessarily
with distinct terms) such that $S(\theta_n)\to S_*$, and among them,
there further exists a subsequence $\{\theta_{n_k}\}$ such that
$(\theta_{n_k},q(\theta_{n_k}))$ converges, and denote $\theta^* =
\lim \theta_{n_k}$. As $F$ and $F^{-1}$ are continuous, and $q$ is
right-continuous and decreasing, $S$ is lower semi-continuous so $S(\theta^*) =S_*$.)

We next show that $(\theta^*,q(\theta^*))$ cannot be strictly dominated in the sense of the relation $\prec$ in \Cref{defi:order}.  Otherwise there exists $(s,q(s)) \succ (\theta^*,q(\theta^*))$. By \Cref{lemm:equivalence}, this implies that
\begin{equation*}
   s+a_s \le \theta^*+a^*, \quad s+b_s \le \theta^*+b^*,
\end{equation*}
where $a^* = F^{-1}(q(\theta^*)/m)$, $b^* = F^{-1}(q(\theta^*))$, $a_s
= F^{-1}(q(s)/m)$, and $b_s = F^{-1}(q(s))$, with at least one strict
inequality among the two. Summing up the two inequalities yields
\begin{align*}
2s+a_s+b_s < 2\theta^*+a^*+b^* = S_*,
\end{align*}
which contradicts the definition of $S_*$.


Below, we show that the $q^*$ curve passes through $(\theta^*,q(\theta^*))$. As noted in the main text, $q^*(\theta^*) \leq q(\theta^*)$ by definition.
 If $q^*(\theta^*) < q(\theta^*)$, we have, using the definition of $q^*$ and the properties of the infimum,
   \begin{equation*}
       \sup_{a\le x - \theta^\prime \le b} q(\theta^\prime)
       \frac{F(x-\theta^*)}{F(x - \theta^\prime)} <  q(\theta^*) \quad
       \text{for some}~\theta'.
   \end{equation*}
   Applying \Cref{lemm:minimizer} with $t=\theta'$ and
   $\theta=\theta^*$, we know that
   \[\bar{q}^*_{\theta'}(\theta^*) \le \sup_{a \le x-\theta' \le b}
   q(\theta')\frac{F(x-\theta^*)}{F(x-\theta')}.
   \]
   Combining these inequalities yields $\bar{q}^*_{\theta'}(\theta^*)
   < q(\theta^*)$, which contradicts the observation that $(\theta^*,
   q(\theta^*))$ cannot be dominated. Therefore, $q^*(\theta^*) = q(\theta^*)$.
\end{proof}



\begin{proof}[Proof of \Cref{lemm:minimizer}]\label{proof:lemm:minimizer}
We consider two cases. For $\theta' < t$, we have
$\bar q_t(\theta') = 1$ and $b = F^{-1}(1)=\infty$.
Therefore,
\begin{align*}
\sup_{x-\theta' \ge a} \bar q_t(\theta')
\frac{F(x-\theta)}{F(x-\theta')}
= \sup_{x-\theta' \ge a}
\frac{F(x-\theta)}{F(x-\theta')}
\ge \lim_{x \to \infty} \frac{F(x-\theta)}{F(x-\theta')}
= 1.
\end{align*}
For $\theta' \ge t$, we have $\bar q_t(\theta') = q(t)$, and
\begin{align*}
&\sup_{a \le x-\theta' \le b}\bar q_t(\theta')\frac{F(x-\theta)}{F(x-\theta')}\\
=& \sup_{a \le x' \le b}\bar q_t(\theta')\frac{F(x'+\theta'-\theta)}{F(x')} && (x' := x-\theta')\\
=& \sup_{F^{-1}(q(t)/m) \le x' \le F^{-1}(q(t))} q(t)\frac{F(x'+\theta'-\theta)}{F(x')} && (\bar q_t(\theta') = q(t)\ \text{for}\ \theta'\ge t)\\
\ge& \sup_{F^{-1}(q(t)/m) \le x' \le F^{-1}(q(t))} q(t)\frac{F(x'+t-\theta)}{F(x')} && (\theta' \ge t,\ F\ \text{increasing, equality can be obtained at $\theta' = t$.
})\\
=& \sup_{a \le x-t \le b}\bar q_t(t)\frac{F(x-\theta)}{F(x-t)} && (x := x' + t).
\end{align*}
This shows $\bar q_t^*(\theta)$ is lower bounded by the expressions in
\Cref{lemm:minimizer}. Now consider
$\theta' = \theta$, so
\begin{align*}
\sup_{a \le x-\theta \le b} \bar q_t(\theta)
\frac{F(x-\theta)}{F(x-\theta)}
= \bar q_t(\theta) \le 1.
\end{align*}
This establishes the equality in \Cref{lemm:minimizer}.

For the special case $\bar{q}_t(t)=1$, we have $\bar{q}_t(t) = 1$ and $b = F^{-1}(q(t)) = \infty$. Then, for $\theta \ge t$, $\sup_{a \le x-t \le b} \bar{q}_t(t)\frac{F(x-\theta)}{F(x-t)} = 1$; for $\theta < t$, we have $\sup_{a \le x-t \le b} \bar{q}_t(t)\frac{F(x-\theta)}{F(x-t)} \ge 1$, and therefore $\min\left\{ 1, \sup_{a \le x-t \le b} \bar{q}_t(t)\frac{F(x-\theta)}{F(x-t)} \right\} = 1$.
\end{proof}

\begin{proof}[Proof of \Cref{prop:lower.bound}]

Consider an arbitrary $\theta < 0$. For
the standard $\BH_q$ procedure, the monotone ratio property ensures
that the supremum in the definition of $q^*$ in \eqref{eq:q-star} is
attained at the boundary $t = F^{-1}(q/m)$. Thus, $q^*(\theta)
= q F(t - \theta) / F(t)$, or equivalently $F(t -
\theta) = q^*(\theta) / m$.
Consider the parameter configuration $\theta_1 = \cdots = \theta_m = \theta$. Under this configuration, the null hypothesis $H_{i,\theta}$ (which is $\theta_i \ge \theta$) is true for all $i = 1, \dots, m$. Therefore, any rejection is a false discovery, so
\begin{align*}
    \FDR(\theta) &= \mathbb{P}(|\calS_{\BH_q}| \ge 1) \\
    &= 1 - \mathbb{P}(X_i > t \text{ for all } i=1,\dots,m) \\
    &= 1 - \left( 1 - F(t - \theta) \right)^m \\
    &=1 - \left( 1 - \frac{q^*(\theta)}{m} \right)^m \\
    &\ge 1 - e^{-q^*(\theta)}.
\end{align*}
This completes the proof.
\end{proof}

\begin{proof}[Proof of \Cref{lemm:equivalence}]
Let $a_1 = F^{-1}(q_1/m), b_1 = F^{-1}(q_1)$ and  $a_2 = F^{-1}(q_2/m), b_2 = F^{-1}(q_2)$.
Then
\begin{align*}
&(\theta_1,q_1)\preceq (\theta_2,q_2)\\
\Longleftrightarrow\;&
\sup_{a_2 \le x-\theta_2 \le b_2}
q_2 \frac{F(x-\theta_1)}{F(x-\theta_2)}
\le q_1 && \text{(\Cref{lemm:minimizer}, $q_1 < 1$)} \\
\Longleftrightarrow\;&
q_2 \frac{F(\theta_2+a_2-\theta_1)}{F(a_2)} \le q_1,
\quad
q_2 \frac{F(\theta_2+b_2-\theta_1)}{F(b_2)} \le q_1
&& \text{(Condition~\eqref{cond:monotone.ratio})} \\
\Longleftrightarrow\;&
mF(\theta_2+a_2-\theta_1) \le q_1,
\quad
F(\theta_2+b_2-\theta_1) \le q_1
&& \text{($F(a_2)=q_2/m$, $F(b_2)=q_2$)} \\
\Longleftrightarrow\;&
F(\theta_2+a_2-\theta_1) \le F(a_1),
\quad
F(\theta_2+b_2-\theta_1) \le F(b_1)
&& \text{($F(a_1)=q_1/m$, $F(b_1)=q_1$)} \\
\Longleftrightarrow\;&
\theta_2+a_2-\theta_1 \le a_1,
\quad
\theta_2+b_2-\theta_1 \le b_1
&& \text{($F$ is continuous and increasing)} \\
\Longleftrightarrow\;&
\theta_2+a_2 \le \theta_1+a_1,
\quad
\theta_2+b_2 \le \theta_1+b_1.
 \end{align*}

For $\prec$, we automatically have $\theta_2+a_2 \le \theta_1+a_1$, $\theta_2+b_2 \le \theta_1+b_1$.
We prove by negation that the two  equalities can not hold
simultaneously.
Otherwise, suppose $\theta_2+a_2 = \theta_1+a_1$ and $\theta_2+b_2 = \theta_1+b_1$.
By Lemma~\ref{lemm:minimizer} and the monotone ratio property,
\begin{align*}
    \bar{q}^*_{\theta_2}(\theta_1) &= \min\left\{ \bar{q}_{\theta_2}(\theta_1), \max\left\{ q_2 \frac{F(\theta_2 + b_2-\theta_1)}{q_2}, q_2 \frac{F(\theta_2 + a_2-\theta_1)}{q_2/m} \right\} \right\} \\
    &= \min\left\{ \bar{q}_{\theta_2}(\theta_1), \max\left\{ F(b_1), m F(a_1) \right\} \right\} \\
    &= \min\left\{ \bar{q}_{\theta_2}(\theta_1), q_1 \right\}.
\end{align*}
If $\theta_1 < \theta_2$, then $\bar{q}_{\theta_2}(\theta_1) = 1 \ge q_1$, so $\bar{q}^*_{\theta_2}(\theta_1) = q_1$.
If $\theta_1 \ge \theta_2$, then $\theta_1 + a_1 = \theta_2 + a_2$ implies $a_1 \le a_2$, so $q_1 \le q_2$. Thus $\bar{q}_{\theta_2}(\theta_1) = q_2 \ge q_1$, and again $\bar{q}^*_{\theta_2}(\theta_1) = q_1$.
In both cases, $\bar{q}^*_{\theta_2}(\theta_1) = q_1$, which
contradicts the assumption $(\theta_1, q_1) \prec (\theta_2, q_2)$
(i.e., $\bar{q}^*_{\theta_2}(\theta_1) < q_1$).
\end{proof}

\begin{proof}[Proof of \Cref{prop:touching.point.heterogeneous.CDF}]
We prove the more general result that allows different $F_i$ by induction on $K:=|\Theta|$.
If $K=1$, \Cref{lemm:minimizer} 
can be easily extended to this general problem and implies
$q^*(\theta)=q(\theta)$ for $\theta \in \Theta$. The details are
omitted.

Now suppose the claim holds for any $K$ jump points, and consider $\Theta=\{\theta_1,\dots,\theta_K,\theta_{K+1}\}$.
By the induction hypothesis, there exists a subset $\Theta_{K}' \subseteq \{\theta_1,\dots,\theta_K\}$ such that the associated curve $q^*_{\Theta_{K}'}$ satisfies
$
q^*_{\Theta_{K}'}(\theta_j)\le q(\theta_j)$, $j=1,\dots,K$,
and touches $q$ at at least one point in $\Theta_{K}'$.
If $q^*_{\Theta_{K}'}(\theta_{K+1}) \le q(\theta_{K+1})$,
then it controls all $K+1$ constraints and touches $q(\theta)$ at at least one point.
Otherwise, define
$\Theta' := \Theta_{K}' \cup \{\theta_{K+1}\}$.
Recall the definition of $q^*$ in \eqref{eq:q-star} where $q^*_{\Theta'}$ is obtained by taking the infimum over one additional $\theta_{K+1}$, then we have $q^*_{\Theta'}(\theta)\le q^*_{\Theta_{K}'}(\theta)$, for all $\theta \in \mathbb R$.
Therefore,
$q^*_{\Theta'}(\theta_j)\le q^*_{\Theta_{K}'}(\theta_j)\le q(\theta_j)$, $j=1,\dots,K$.
At the new point $\theta_{K+1}$, the term in \eqref{eq:q-star.heterogeneous.CDF} corresponding to $\theta'=\theta_{K+1}$ equals
\[
\sup_{i\in[m]}
\sup_{a_i \le x-\theta_{K+1} \le b_i}
q(\theta_{K+1})\frac{F_i(x-\theta_{K+1})}{F_i(x-\theta_{K+1})}
=
q(\theta_{K+1}),
\]
and by the inductive hypothesis $q^*_{\Theta_{K}'}(\theta_{K+1})>q(\theta_{K+1})$.
Therefore,
\[
q^*_{\Theta'}(\theta_{K+1})
=
\min\left\{q^*_{\Theta_{K}'}(\theta_{K+1}),\, q(\theta_{K+1})\right\}
=
q(\theta_{K+1}).
\]
Thus $q^*_{\Theta'}$ controls all $K+1$ constraints and touches $q(\theta)$ at $\theta_{K+1}$.
\end{proof}

\end{document}